\documentclass[11pt,reqno]{amsart}
\usepackage{graphicx,amsmath,amsthm,verbatim,amssymb,eucal} %eucal is "Euler" mathcal for the script N denoting an abelian network
\usepackage{bbm}
\usepackage{hyperref}
\hypersetup{colorlinks}
\newcommand{\arxiv}[1]{{\tt \href{http://arxiv.org/abs/#1}{arXiv:#1}}}

\newcommand{\Set}[2]{\left\{ #1 \, \left| \; #2 \right. \right\}}

\newcommand{\modnospace}[1]{\; (\mbox{mod }{#1})}

\newcommand{\old}[1]{}
\newcommand{\moniker}[1]{{\em (#1)}}

\newcommand{\silentsubsec}{\subsection*}

\newtheorem{theorem}{Theorem}[section]

\newtheorem{proc}[theorem]{Procedure}
\newtheorem{lemma}[theorem]{Lemma}
\newtheorem{corollary}[theorem]{Corollary}

\theoremstyle{remark}
\newtheorem*{remark}{Remark}
\newtheorem*{example}{Example}

\numberwithin{counter}{section}

\theoremstyle{definition}
\newtheorem{definition}[theorem]{Definition}

\def\isom{\simeq}

\def\mm{\mathbf{m}}

\def\qq{\mathbf{q}}
\def\rr{\mathbf{r}}

\def\uu{\mathbf{u}}

\def\xx{\mathbf{x}}
\def\yy{\mathbf{y}}
\def\zz{\mathbf{z}}
\def\00{\mathbf{0}}
\def\mm{\mathbf{m}}
\def\kk{\mathbf{k}}

\def\lcm{\mathrm{lcm}}

\def\Proc{\CMcal{P}}  % regular mathcal
\def\ProcR{\CMcal{R}}
\def\Net{{\mathcal{N}}}  % "Euler" mathcal

\def\Rec{\mathrm{Rec\,}}
\def\Crit{\mathrm{Crit\,}}
\def\acts{\mathop{\triangleright}}
\def\Acts{\mathop{\hspace{1pt} \triangleright \hspace{0.5pt} \triangleright}}
\def\Actsinline{\Acts}
\def\End{\mathrm{End\, }}
\def\zero{\mathbf{0}}
\def\one{\mathbf{1}}
\def\Sa{\mathcal{S}}

\def\M{M}
\def\basis{1}
\def\Rotor{{\tt Rotor}}
\def\Sand{{\tt Sand}}

\def\N{\mathbb{N}}
\def\Z{\mathbb{Z}}
\def\Q{\mathbb{Q}}

\def\EE{\mathbb{E}}
\def\eps{\epsilon}

\def\Prob{\Pr}

\begin{document}

\title[Abelian Networks]{Abelian Networks III. The Critical Group}

\author[Benjamin Bond and Lionel Levine]{Benjamin Bond and Lionel Levine}

\address{Benjamin Bond, Department of Mathematics, Stanford University, Stanford, California 94305. {\tt\url{http://stanford.edu/~benbond}}}

\address{Lionel Levine, Department of Mathematics, Cornell University, Ithaca, NY 14853. {\tt \url{http://www.math.cornell.edu/~levine}}}

\date{October 29, 2015}
\keywords{abelian distributed processors, asynchronous computation, burning algorithm, chip-firing, commutative monoid action, Laplacian lattice, sandpile group, script algorithm}
\subjclass[2010]{%%update on arxiv
05C25, % Graphs and abstract algebra 
05C50, % Graphs and linear algebra (matrices, eigenvalues, etc.)
20M14, % Commutative semigroups
20M35, % Semigroups in automata theory, linguistics, etc.
68Q10, % Modes of computation (nondeterministic, parallel, interactive, probabilistic, etc.)
%37B15, % Cellular automata
}

\begin{abstract}
The critical group of an abelian network is a finite abelian group that governs the behavior of the network on large inputs. It generalizes the sandpile group of a graph.
We show that the critical group of an irreducible abelian network acts freely and transitively on recurrent states of the network. We exhibit the critical group as a quotient of a free abelian group by a subgroup containing the image of the Laplacian, with equality in the case that the network is rectangular. We generalize Dhar's burning algorithm to abelian networks, and estimate the running time of an abelian network on an arbitrary input up to a constant additive error.
\end{abstract}

\maketitle
%\newpage
%\tableofcontents

\section{Introduction}
\label{s.intro}

Associated to a finite connected graph $G$ with marked vertex $s$ is a finite abelian group called its \emph{sandpile group}. This group arose independently in three different fields: the Neron model of a curve in arithmetic geometry \cite{Lor89,Lor91}, the abelian sandpile model in statistical physics \cite{Dha90}, and discrete potential theory on graphs \cite{BDN97,Big99}. 
In this paper we are going to exhibit the sandpile group as an instance of a more general construction, the critical group $\Crit \Net$ of an abelian network $\Net$. 

The abelian sandpile model (also known as chip-firing) is a discrete dynamical system that redistributes chips
on the vertices of $G$ according to certain moves called topplings. From our perspective, the abelian sandpile model is rather special for a few reasons:
\begin{enumerate}
\item The total number of chips is conserved by toppling.
\item One vertex $s$ plays a distinct role: it can receive chips but cannot topple.
\item Chips are regarded as indistinguishable.
\end{enumerate}
A goal of this paper is to show that the basic theory of the sandpile group can be adapted 
%with only minor modifications 
to a setting where the first two properties are removed and the third is relaxed.
 
In our setting the chips are replaced by letters passed between ``processors'', which are automata located at the vertices of $G$.
%Sandpile networks occupy a tiny corner of the vast universe of automata networks, but their theory carries over to the galaxy of abelian networks.
Each processor has its own state space and input alphabet.
Depending on its internal state and the letter it reads, a processor may pass zero, one or more letters to one or more neighboring processors. For example, in the abelian sandpile model, the processor at a vertex of degree $d$ has state space $\{0,1,\ldots,d-1\}$ and a one-letter input alphabet.  Whenever it reads a letter, the processor increments its state by one modulo $d$; if its new state is $0$ then it passes one letter to each neighboring processor (and otherwise it passes nothing). 
Then, regardless of its starting state, reading $d$ letters returns the processor to its starting state and exactly one letter has been passed to each neighbor. This is precisely the toppling move of the sandpile model.

Now consider varying the state transition rule or the message passing rule of one or more processors. In general, the resulting network of automata need not have a conserved quantity, which removes item (1) from the above list. The special vertex $s$ in item (2) may be replaced by a condition that is both weaker and more symmetric, namely that the network halts on all inputs. The input alphabet of a processor need not consist of just one letter, which allows us to relax item (3). In order to generalize the theory of the sandpile group, items (1)--(3) are not essential. Rather, what is essential is that each processor is \emph{abelian}. 

\subsection{Abelian processor axioms}
\label{s.axioms}

An \emph{abelian network} is a collection of abelian processors indexed by the vertices of a directed graph.  The formal definition of an \emph{abelian processor} appears in Section~\ref{s.abelian}. In words, an abelian processor is an automaton with output, satisfying two axioms: 
\begin{itemize}
\item[(i)] For any initial state and any two input words that are equal up to permutation, the resulting final states are equal. \smallskip
\item[(ii)] For any initial state and any two input words that are equal up to permutation, the resulting output words are equal up to permutation.
\end{itemize}
For example, suppose a abelian processor starts in some state $q$ and reads the input word $ab$, and that as a result, it ends in some state $q'$ and outputs the word $ccddc$.  If instead we were to input the word $ba$ to the same processor in the same starting state $q$, then axiom (i) says the final state will still be $q'$; and axiom (ii) says the resulting output word need not be $ccddc$, but must have exactly $5$ letters: $3$ $c$'s and $2$ $d$'s.

Deepak Dhar \cite{Dha99} proposed abelian networks as models of self-organized criticality in physics, generalizing the abelian sandpile model of \cite{BTW87, Dha90}. From the point of view of computer science, abelian networks are an interesting class of automata networks because they can compute asynchronously: an abelian network produces the same final output in the same number of steps regardless of the order of events at individual nodes of the network \cite{part1}. 

Examples of abelian networks are surveyed in \cite{part1}. They include sandpile and rotor networks and their non-unary cousins, oil and water networks and abelian mobile agents. Besides the sandpile groups, critical groups of abelian networks have been studied in a few other particular cases: rotor networks \cite{PDDK96} and height-arrow networks \cite{DR04}. 
The group $\Crit \Rotor(G,s)$ associated to a rotor network with a sink turns out to be isomorphic to the sandpile group, $\Crit \Sand(G,s)$ \cite{PDDK96,LL09}.  We will see this isomorphism 
as a case of a more general phenomenon: Homotopic abelian networks have isomorphic critical groups (Corollary~\ref{c.homotopy}).

The sandpile group has several different constructions. Babai and Toumpakari \cite{BT10} realized the sandpile group as the minimal ideal of a commutative monoid. Their approach, developed further in \cite{C+13}, is well-suited for generalizing to abelian networks. In \textsection\ref{s.monoid} we review the small amount of monoid theory we will need.
In \textsection\ref{s.abelian} we recall the definition of an abelian network and relevant results from \cite{part1,part2}, and prove a few basic lemmas including a local-to-global principle for irreducibility, Lemma~\ref{l.localglobalirreducible}.

\subsection{Main results}

The setting for all of our main results is a finite irreducible abelian network $\Net$ that halts on all inputs.
In \textsection\ref{s.recurrent} we define the critical group $\Crit \Net$ and the set $\Rec \Net$ of recurrent states, and show that the former acts freely and transitively on the latter.  In \textsection\ref{s.markov} we consider a Markov chain defined by sending random input to the network; we relate the algebraic and probabilistic definitions of ``recurrent,'' and show that the stationary distribution of the chain is uniform on recurrent states. In \textsection\ref{s.expectedtime} we find the expected time for $\Net$ to halt on a given input to a uniform random recurrent state. In \textsection\ref{s.genrel} we give generators and relations for the critical group, and in \textsection\ref{s.order} we compute its order.

In \textsection\ref{s.timetohalt} we estimate the time for $\Net$ to halt on an arbitrary input, up to a constant additive error.

In \textsection\ref{s.large} we give an efficient test for whether a given state of $\Net$ is recurrent, generalizing Dhar's burning algorithm \cite{Dha90} and Speer's script algorithm \cite{Spe93}. We define a \emph{burning element} to be an input that returns the network to its initial state if and only if that state is recurrent. We show that any integer vector satisfying a certain set of linear inequalities is a burning element. In \textsection\ref{s.finding} we give an efficient way to find such a vector.

\subsection{Review of monoid actions}
\label{s.monoid}

Here we recall a few facts about actions of finite commutative monoids. Most of this material is implicit in the semigroup literature \cite{Gre51,Sch57,Gri01,Ste10}. The proofs are straightforward and can be found in \cite{part2}. For refinements and a generalization to certain infinite semigroups, see Grillet \cite{Gri07}. 

Let $M$ be a commutative monoid, $X$ a set and $\mu : M \times X \to X$ a monoid action. Commutativity of $M$ implies that the relation $\sim$ on $X$ defined by
	\begin{equation} \label{e.mutualaccess}  x \sim x' : \qquad \exists m,m' \in M \text{ such that } mx=m'x' \end{equation}
is an equivalence relation. We say that $\mu$ is \emph{irreducible} if $\sim$ has only one equivalence class.

For the rest of this section we assume that $M$ is finite and $\mu$ is irreducible. By finiteness, every $m \in M$ has an idempotent power: $m^j = m^{2j}$ for some $j \geq 1$. The \emph{minimal idempotent} $e := \prod_{f \in M, \, ff=f} f$ is the unique idempotent accessible from all of $M$ (that is, $ee=e$ and $e \in Mm$ for all $m \in M$).

\begin{lemma}
%Let $M$ be a finite commutative monoid. Then
$eM$ is an abelian group with identity element $e$.
\end{lemma}

Indeed, the existence of inverses is immediate from $e \in Mm$.

\begin{lemma}
\moniker{Recurrent Elements Of An Irreducible Monoid Action}
\label{l.recurrent}
The following are equivalent for $x \in X$:
\begin{enumerate}
\item $x \in My$ for all $y \in X$
\item $x \in M(mx)$ for all $m \in M$
\item $x \in mX$ for all $m \in M$
\item $x \in eX$
\item $x = ex$
\end{enumerate}
\end{lemma}

An element $x \in X$ is called \emph{recurrent} if it satisfies the equivalent conditions of Lemma~\ref{l.recurrent}. To explain this terminology, suppose we are given a probability distribution $\alpha$ on $M$ such that $\{m\in M \,:\, \alpha(m)>0\}$ generates $M$ as a monoid, and consider the Markov chain on $X$ that transitions from $x$ to $m x$ with probability $\alpha(m)$.  The states $x$ satisfying the equivalent conditions of Lemma~\ref{l.recurrent} are precisely those visited infinitely often by the Markov chain (see \textsection\ref{s.markov}).

We say that $m \in M$ \emph{acts invertibly} on a subset $Y$ of $X$ if the map $y \mapsto my$ is a permutation of $Y$.

\begin{lemma}
%note: does not require irreducible
\label{l.actsinvertibly}
Every $m \in M$ acts invertibly on $eX$.
\end{lemma}

We say that a monoid action $\mu : M\times X \to X$ is \emph{faithful} if there do not exist distinct elements $m,m' \in M$ such that $mx=m'x$ for all $x \in X$.  
%Equivalently, $\mu$ defines an injection of $M$ into the monoid $\End X$ of all self-maps of $X$.

Let $G$ be a group with identity element $e$.  Recall that a group action $G \times Y \to Y$ is called \emph{transitive} if $Gy = Y$ for all $y \in Y$, and is called \emph{free} if for all $g \neq e$ there does not exist $y \in Y$ such that $gy=y$.  If the action is both transitive and free, then for any two elements $y,y' \in Y$ there is a unique $g \in G$ such that $gy=y'$; in particular, $\# G = \# Y$.

\begin{theorem}
\label{t.monoidtogroup}
\moniker{Group Actions Arising From Monoid Actions}
Let $M$ be a finite commutative monoid and $\mu: M \times X \to X$ an irreducible action.
The restriction of $\mu$ to $eM \times eX$ is a transitive group action
	 \[ e\mu: eM \times eX \to eX. \]
In addition, if $\mu$ is faithful, then $e\mu$ is free.
\end{theorem}

\section{Abelian networks}
\label{s.abelian}

We now recall the definition and basic properties of abelian networks, refering the reader to \cite{part1,part2} for details.

\subsection{Abelian processors}

Let $Q$ be a set of ``states'' and $\End(Q)$ the monoid of all set maps $Q \to Q$ with the operation of composition. An \emph{abelian processor} with input alphabet $I$ and state space $Q$ is a collection of maps $(t_a)_{a \in I}$ where each $t_a \in \End(Q)$, such that $t_a t_b = t_b t_a$ for all $a,b \in I$. This commutativity implies axiom (i) from the introduction.

\subsection{Abelian processors with output}

So far an abelian processor can take input and change state. Next we will enable it to send output, so that it can pass messages to other processors in a network. If a processor has several neighbors in the network, we will allow it to pass a different message to each. So an abelian processor may have several output alphabets (one for each neighbor).

Let $U$ be a set of ``output feeds''. An \emph{abelian processor with output alphabets} $(A_u)_{u \in U}$ 
%(in our application, the $A_u$ will be the input alphabets of the neighboring processors) 
has in addition to the state transition maps $t_a$ a message passing function $o_a^u : Q \to A_u^*$ for each $a\in I$ and $u \in U$.  Here $A^*$ denotes the free monoid of all finite words in an alphabet $A$. These functions are required to satisfy a commutativity condition: namely, if two input words $w,w' \in I^*$ are equal up to permutation, then for each $u \in U$ the resulting output words in $A_u^*$ are equal up to permutation. This condition is axiom (ii) from the introduction.

\subsection{Abelian networks}

An \emph{abelian network} on a directed graph $G=(V,E)$ with alphabet $A = \sqcup_{v \in V} A_v$ and state space $Q = \prod_{v \in V} Q_v$ is a collection $(\Proc_v)_{v \in V}$, where each $\Proc_v$ is an abelian processor with input alphabet $A_v$, state space $Q_v$ and output alphabets $(A_u)_{(v,u) \in E}$ indexed by the outgoing edges from $v$. When $\Proc_v$ in state $q \in Q_v$ processes letter $a \in A_v$, it transitions to state $t_a(q)$ and sends the message $o_a^u(q)$ to each neighboring processor $\Proc_u$.  Note that $o_a^u(q)$ might be the empty word, which signifies that no message is sent.

The total state of an abelian network $\Net = (\Proc_v)_{v \in V}$ is described by an element $\qq \in Q$ indicating the internal states of all processors, together with a vector $\xx \in \Z^A$ indicating how many letters of each type are waiting to be processed.  We use the notation $\xx.\qq$ for this pair (in which the decimal point is meant to suggest that the states $\qq$ represent ``fractional letters'' that have not yet been output). Note that $A$ is a disjoint union, so each letter belongs to the input alphabet of a unique processor.

\subsection{Sandpiles, rotor networks, toppling networks}
\label{s.examples}

These will be our running examples of abelian networks. Let $G=(V,E)$ be a finite directed graph and $s \in V$ a vertex such that from every other vertex $v \in V$ there is a directed path from $v$ to $s$. Each processor $\Proc_v$ has alphabet $A_v = \{v\}$ and state space $Q_v = \Z/r_v \Z$ for a positive integer $r_v$ called the \emph{threshold} of $v$. The state transition is $t_v (q) = q+1$ (mod $r_v$). 

It remains to describe the message passing.
In the networks $\Sand(G,s)$ and $\Rotor(G,s)$, processor $\Proc_s$ is a \emph{sink} (which means it has just one state and never passes any messages) and we take $r_v = d_v$, the outdegree of $v$, for all $v \neq s$.
In the sandpile network $\Sand(G,s)$, whenever processor $\Proc_v$ transitions from state $d_v - 1$ to state $0$ it passes $d_v$ letters: one letter $u$ along each outgoing edge $(v,u)$. The message passing for the rotor network $\Rotor(G,s)$ is specified by fixing an ordering $e_1,\ldots,e_{d_v}$ of the outgoing edges from $v$. Whenever it transitions from state $q$ to $q+1$ (mod $d_v$), processor $\Proc_v$ passes exactly one letter, along the edge $e_{q+1}$.

\emph{Toppling networks} are a generalization of sandpiles, where we allow $r_v \neq d_v$; however, when processor $\Proc_v$ in a toppling network transitions from state $r_v - 1$ to state $0$, it passes $d_v$ letters just as in a sandpile network: one letter $u$ along each outgoing edge $(v,u)$. The importance of the sink $s$ in $\Sand(G,s)$ and $\Rotor(G,s)$ is to ensure the network halts on all inputs. Depending on the thresholds $r_v$, a toppling network may halt on all inputs even if no vertex is a sink.

Rotor and toppling networks are \emph{unary}, meaning that each alphabet $A_v$ is a singleton. See \cite{part1} for two examples of non-unary abelian networks, oil and water and abelian mobile agents.

\subsection{Executions}

An \emph{execution} is a finite word $w \in A^*$. It prescribes an order in which letters are to be processed. We write $\pi_w(\xx.\qq)$ for the result of executing $w$ starting from $\xx.\qq$; this is another pair $\xx'.\qq'$ that can be computed using the state transitions and message passing functions. It is important to note that some coordinates of $\xx'$ may be negative! For example, if $w$ consists of a single letter $a$ and $\xx_a=0$, and processing $a$ does not cause any letters $a$ to be passed, then $\xx'_a = -1$. The interpretation is that a processor was instructed to process letter $a$ even though no letter $a$ was present; the processor follows the instruction and keeps track of the ``debt'' that it is owed one letter $a$. Note however that messages passed from one processor to another are always nonnegative.

The axioms (see \S\ref{s.axioms}) of an abelian processor imply that $\pi_w(\xx.\qq)$ depends only on the vector $|w| \in \N^A$ where $|w|_a$ is the number of letters $a$ in $w$. We write $\pi_w$ and $\pi_{|w|}$ interchangeably.
%(here it is crucial that $\xx\in \Z^A$ rather than $\N^A$).
%By definition $\pi_w$ is the composition $\pi_{a_1} \circ \cdots \circ \pi_{a_r}$. Each $\pi_a$ has three effects: change the internal state $\qq_v$ to $t_a \qq_v$,  where $v$ is the unique vertex such that $a \in A_v$; decrement $\xx_a$ by one; and increment each coordinate $\xx_b$ by the number of letters $b$ passed (as specified by the message passing functions and the internal states of the processors).
Writing $w = a_1 \cdots a_r$ and $\pi_{a_1\cdots a_i}(\xx.\qq) = \xx^i.\qq^i$, we say that $w$ is \emph{legal} for $\xx.\qq$ if $\xx^{i-1}_{a_i} \geq 1$ for all $i=1,\ldots,r$. We say that $w$ is \emph{complete} for $\xx.\qq$ if $\xx^r \leq \zero$ (inequalities between vectors are coordinatewise). In words, a legal execution is one that incurs no ``debts,'' and a complete execution is one that removes all letters from the network. Note that if $w$ is both legal and complete, then $\xx^i \geq \zero$ for $i=1,\ldots,r-1$ and $\xx^r = \zero$. 

 The least action principle \cite[Lemma 4.3]{part1} says that if $w$ is any legal execution for $\xx.\qq$ and $w'$ is any complete execution for $\xx.\qq$, then $|w| \leq |w'|$. It follows that $|w|$ is the same for all complete legal executions of $\xx.\qq$. 
If there exists a complete legal execution for $\xx.\qq$ then we say that $\Net$ \emph{halts} on input $\xx.\qq$. 
% can leave out ``legal'' here but I think it's clearer to keep it in.
In this case the \emph{odometer} $[\xx.\qq]$ is defined as
	\[ [\xx.\qq] := |w| \]
where $w$ is any complete legal execution for $\xx.\qq$. The odometer is a vector in $\N^A$ whose $a$th coordinate is the total number of letters $a$ processed.

The final state $\qq^r_v$ of each processor $\Proc_v$ can be determined from its initial state $\qq_v$ and the odometer coordinates $[\xx.\qq]_a$ for $a \in A_v$, namely
	\begin{equation} \label{e.odomdeterminesfinalstate} \qq^r_v 
	%= \prod_{i \,:\, a_i \in A_v} t_{a_i} \qq_v 
	= \left( \prod_{a \in A_v} t_a^{[\xx.\qq]_a} \right) \qq_v. \end{equation}
where the product denotes composition of the commuting maps $t_a$.  To make its dependence on $\xx$ explicit we will use the notation 
	$\xx \Acts \qq$
for the final state $\qq^r$. Like the odometer $[\xx.\qq]$, the state $\xx \Acts \qq$ depends only on $\xx$ and $\qq$ and not on the choice of complete legal execution.
%	\[ \pi_{[\xx.\qq]}(\xx.\qq) = \zero.(\xx \Acts \qq). \]

\subsection{The global action}

Let $\Net$ be an abelian network that halts on all inputs: that is, $\xx.\qq$ has a finite complete legal execution for all $\xx \in \N^A$ and all $\qq \in Q$. We will see shortly that $(\xx,\qq) \mapsto \xx \Acts \qq$ defines a monoid action $\N^A \times Q \to Q$. This \emph{global action} is the main object of interest in the present paper: we will use this action to define the critical group in \textsection\ref{s.critical} and characterize the recurrent states of this action in \textsection\ref{s.burning}. We have chosen the notation $\Acts$ to distinguish the global action from the \emph{local action} $\acts$ of \cite{part2}; below we recall how $\acts$ is defined and relate the two actions.

\begin{lemma}
\label{l.partialexec}
If $w$ is a legal execution from $\xx.\qq$ to $\yy.\rr$, then 
	\[ \xx \Acts \qq = \yy \Acts \rr \] 
and
	\[ [\xx.\qq] = |w| + [\yy.\rr]. \]
\end{lemma}

\begin{proof}
Let $w'$ be a complete legal execution for $\yy.\rr$. Then the concatenation $ww'$ is a complete legal execution for $\xx.\qq$.
\end{proof}

The next lemma verifies that $\Acts$ defines a monoid action of $\N^A$ on $Q$.

\begin{lemma}
\label{l.odomsplit}
For $\xx,\yy \in \N^A$ and $\qq \in Q$ we have 
	\[ (\xx+\yy)\Acts \qq = \yy \Acts (\xx \Acts \qq) \] 
and
	\[ [(\xx+\yy).\qq] = [\xx.\qq] + [\yy.(\xx \Acts \qq)]. \]
\end{lemma}

\begin{proof}
If $w$ is a complete legal execution for $\xx.\qq$, then $w$ is a legal execution from $(\xx+\yy).\qq$ to $\yy.(\xx \Acts \qq)$.  Since $|w| = [\xx.\qq]$ the result follows from Lemma~\ref{l.partialexec}. 
\end{proof}

\subsection{The local action}

In \cite{part2} we defined a monoid action $\acts$ of $\N^A$ on $\Z^A \times Q$, 
%by the rule ``add letters and process each added letter once:'' 
	\[ \xx \acts (\yy.\qq) = \pi_\xx((\xx+\yy).\qq). \]
%In words, $\xx \acts (\yy.\qq)$ is the state $\yy'.\qq'$ obtained from $\yy.\qq$ by adding $\xx_a$ letters $a$ to the network for each $a \in A$ and processing each added letter once.
This is called the \emph{local action} because each processor $\Proc_v$ processes only the letters that were added at $v$ (namely $\xx_a$ letters $a$ for each $a \in A_v$) and not any additional letters passed from other processors. We write $\xx \acts \qq$ as a shorthand for $\xx \acts (\zero.\qq)$.

The next lemma relates the local and global actions.

\begin{lemma}
\label{l.partialexec2}
If $\xx \acts \qq = \yy.\rr$, then $\xx \Acts \qq = \yy \Acts \rr$ and $[\xx.\qq] = \xx + [\yy.\rr]$.
\end{lemma}

\begin{proof}
If $\xx \acts \qq = \yy.\rr$ then there is a legal execution $w$ from $\xx.\qq$ to $\yy.\rr$ with $|w|=\xx$, so the result follows from Lemma~\ref{l.partialexec}.
\end{proof}

% could explain the following using l.partialexec2
One way to compute the global action $\Acts$ is by iterative application of $\acts$: given $\xx \in \N^A$ and $\qq \in Q$, set $\xx_0 = \xx$ and $\qq_0=\qq$ and
	\begin{equation} \label{e.parallelupdate} \xx_n.\qq_n = \xx_{n-1} \acts \qq_{n-1} \end{equation}
for $n \geq 1$.  This amounts to making a particular choice of execution, called \emph{parallel update}. In parallel update the execution occurs in rounds $n=1,2,\ldots$. At the beginning of round $n$ there are $(\xx_n)_a$ letters $a$ waiting to be processed for each $a \in A$. During round $n$ we execute a word $w_n$ with $|w_n| = \xx_n$. 
The concatenation $w = w_1 w_2 \cdots$ is a legal execution.  Since $\Net$ halts on all inputs, every legal execution is finite, so there is some $N$ such that all words $w_n$ for $n > N$ are empty.  Then $\xx_n = \zero$ for all $n \geq N$, and $w = w_1 \cdots w_N$ is a finite complete legal execution. In particular, the final state is given by
	 \[ \xx \Acts \qq = \qq_{N} \]
and the odometer is
	\[ [\xx.\qq] = |w| = \sum_{n =1}^N \xx_n. \]

We record a few more identities to be used later.  We extend the domain of $\Actsinline$ to $\N^A \times Q$ by defining for $\yy \in \N^A$
	\[ \xx \Acts (\yy.\qq) := (\xx+\yy) \Acts \qq. \]

\begin{lemma}
\label{l.odometertrick}
Given $\xx \in \N^A$ and $\qq \in Q$, let $\kk = [\xx.\qq]$. Then 
\begin{enumerate}
\item[(i)] $\pi_\kk(\xx.\qq) = \zero.(\xx \Acts \qq)$
\item[(ii)] $\kk \acts (\xx.\qq) = \kk.(\xx \Acts \qq)$
\item[(iii)] $\xx \Acts (\yy \acts \qq) = (\xx+\yy) \Acts \qq$
\end{enumerate}
\end{lemma}

\begin{proof}
\old{
Write $s_n = \pi_{\kk_n}(\xx_n.\qq_n)$, where $\kk_n = \sum_{j \geq n} \xx_j$ and $\xx_n.\qq_n$ is defined by \eqref{e.parallelupdate}. Then using $\pi_{\kk_n} = \pi_{\kk_{n+1}} \circ \pi_{\xx_n}$ we have for all $n \geq 0$
	\[ s_n %= \pi_{\kk_{n+1}} \pi_{\xx_n}(\xx_n.\qq_n) 
	= \pi_{\kk_{n+1}} (\xx_n \acts \qq_n) = s_{n+1}. \]
Hence $s_0 = s_N$, which proves part (i).
}
Part (i) follows from the definition of $\xx \Acts \qq$.

Part (ii) follows from (i) since
%using \cite[Lemma~\ref{l.piprops}(ii)]{part2}: 
\[ \kk \acts (\xx.\qq) = \pi_\kk((\xx+\kk).\qq) = \kk.(\xx \Acts \qq). \]

To prove part (iii), both states $\xx \Acts (\yy \acts \qq)$ and $(\xx + \yy) \Acts \qq$ are the result of performing a complete legal execution for $(\xx+\yy).\qq$, and any two complete legal executions for $(\xx+\yy).\qq$ result in the same final state.
\end{proof}

\subsection{Local monoids}
\label{s.localmonoid}

The \emph{transition monoid} of an abelian processor with state space $Q$ and transition maps $t_a : Q \to Q$ is the submonoid $M = \langle t_a \rangle_{a \in A} \subseteq \End(Q)$, where $\End(Q)$ is the monoid of all set maps $Q \to Q$ with the operation of composition.  Since $M$ is defined as a submonoid of $\End(Q)$ it has a faithful monoid action $M \times Q \to Q$.
%The processor is called \emph{irreducible} if the defining action $M \times Q \to Q$ is irreducible (\textsection\ref{s.monoid}).

Each processor in an abelian network $\Net = (\Proc_v)_{v \in V}$ has a transition monoid 
	\begin{equation} \label{e.localmonoid} M_v := \langle t_a \rangle_{a \in A_v} \subset \End(Q_v). \end{equation}
We call this $M_v$ the \emph{local monoid} at $v$.
The product $\prod_{v \in V} M_v$ acts coordinatewise on $Q = \prod_{v \in V} Q_v$. To relate this action to the global and local actions defined above, let $t_v : \N^{A_v} \to M_v$ be the monoid homomorphism sending basis elements $1_a$ to the commuting generators $t_a$, and (recalling $A = \sqcup A_v$) write 
	\[ t : \N^A \to \prod_{v \in V} M_v \] 
for the Cartesian product of the maps $t_v$.  Each $t_v$ is surjective, so $t$ is surjective.  Equations \eqref{e.odomdeterminesfinalstate} (there is one equation for each $v \in V$) can be written more succinctly as the single equation
	\begin{equation} \label{e.tglobal} \xx \Acts \qq = t([\xx.\qq])\qq. \end{equation}
To relate $t$ to the local action, note that if $\xx \acts \qq = \yy.\rr$ then $\rr = t(\xx)\qq$.

\subsection{Global monoid}

If $\Net$ halts on all inputs, then we can view the entire network as a single abelian processor (see \cite[Lemma~4.7]{part1}) with input alphabet $A=\sqcup A_v$ and state space $Q=\prod Q_v$. In this case $\Net$ has a \emph{global monoid}, defined by\	
	\begin{equation} \label{e.globalmonoid} M := \langle \tau_a \rangle_{a \in A} \subset \End(Q) \end{equation}
where $\tau_a(\qq) := 1_a \Acts \qq$. Note that $M$ is not the same as the product of local monoids $\prod M_v$ of \eqref{e.localmonoid}: the local monoids depend only on the state transition maps $t_a$, but $M$ depends also on the message passing functions (because $\Acts$ does). 

Write
	\begin{equation} \label{e.taudef} \tau : \N^A \to M \end{equation}
for the monoid homomorphism sending basis elements $1_a$ to generators $\tau_a$. By Lemma~\ref{l.odomsplit} we have $\xx \Acts \qq = \tau(\xx) \qq$. 

If $\Net$ is a finite abelian network (that is, $V$ is finite, and the alphabet $A_v$ and the state space $Q_v$ of each processor are finite) then 
%$Q$ is finite and hence 
$M$ is a finite commutative monoid. In this case we denote by $e$ the minimal idempotent of $M$ (\textsection\ref{s.monoid}). We will use the following property of $e$ repeatedly.

\begin{lemma} \label{l.largee}
For any $\xx \in \N^A$ there exists $\zz \geq \xx$ such that $\tau(\zz)=e$. \end{lemma}

\begin{proof} Since $e$ is accessible from all of $M$ we have $\tau(\xx)m = e$ for some $m \in M$. Since $\tau$ is surjective we have $m = \tau(\yy)$ for some $\yy \in \N^A$. Now take $\zz = \xx+\yy$.
\end{proof} 

\subsection{Locally irreducible implies globally irreducible}

An abelian processor with transition monoid $M$ and state space $Q$ is called \emph{irreducible} if the defining action $M \times Q \to Q$ is irreducible (\textsection\ref{s.monoid}). Next we prove a local-to-global principle for irreducibility.
	
\begin{lemma}
\label{l.localglobalirreducible}
Let $\Net = (\Proc_v)_{v \in V}$ be an abelian network that halts on all inputs.  If each processor $\Proc_v$ is irreducible, then $\Net$ is irreducible.
% Note we do not need any connectivity assumption on G.
\end{lemma}

\begin{proof}
Let $\qq,\qq' \in Q$.  For each $v \in V$, since $\Proc_v$ is irreducible, 
%by Lemma~\ref{l.chainshortening} 
there exist $m_v,m'_v \in M_v$ such that $m_v \qq_v = m'_v \qq'_v$.  Since $t$ is surjective we can choose $\xx, \xx' \in \N^A$ with $t(\xx)=m$ and $t(\xx')=m'$.  Write
	\[ \xx \acts \qq = \yy.\rr, \qquad \xx' \acts \qq' = \yy'.\rr' \]
where $\rr=\rr'$ since $\rr_v= m_v \qq_v = m'_v \qq'_v = \rr'_v$ for all $v \in V$.  Then by Lemma~\ref{l.odometertrick}(iii),
	\begin{align*} (\xx + \yy') \Acts \qq &= \yy' \Acts (\xx \acts \qq) \\
	&= \yy' \Acts (\yy.\rr) \\
	&= (\yy + \yy') \Acts \rr \\
	&= \yy \Acts (\yy'.\rr) \\
	&= \yy \Acts (\xx' \acts \qq') \\
	&= (\xx'+\yy) \Acts \qq'.
	\end{align*}
Hence $\tau(\xx+\yy')\qq = \tau(\xx'+\yy)\qq'$, so the global action $M\times Q \to Q$ is irreducible.
\end{proof}

In light of Lemma~\ref{l.localglobalirreducible} we will drop ``locally'' from ``locally irreducible'' when referring to an abelian network that halts on all inputs.

\subsection{Total kernel and production matrix}
\label{s.KP}

In \cite{part2} we used the local action to associate two basic algebraic objects to an irreducible abelian network, the \emph{total kernel} $K$ and \emph{production map} $P: K \to \Z^A$.  Together $K$ and $P$ constitute a kind of coarse-grained description of an abelian network: They do not specify the network in full detail, but they capture its ``large scale'' features (for instance, the asymptotic behavior on large inputs). 
In this paper we will see that many properties of interest depend only on $K$ and $P$.

Let $e_v$ be the minimal idempotent of the local monoid $M_v$. A state $\qq \in Q$ is called \emph{locally recurrent} if $\qq_v \in e_v Q_v$ for all $v \in V$.  By Lemma~\ref{l.actsinvertibly}, each $m \in M_v$ acts invertibly on $e_v Q_v$, so we have a group action of $\Z^{A_v}$ on $e_v Q_v$.
The total kernel is defined as 
	\[ K = \prod_{v \in V} K_v \subset \Z^A \] 
where $K_v$ is the kernel of the action $\Z^{A_v} \times e_v Q_v \to e_v Q_v$.  For example, in the case of a sandpile network (\S\ref{s.examples}) we have $e_v Q_v = Q_v = \{0,1,\ldots,d_v-1\}$ (all states are locally recurrent) and $\Z^{A_v} = \Z$ acts by $q \mapsto q+1 \modnospace{d_v}$. So in this case $K_v = d_v \Z$.

If $\Net$ is a finite abelian network, then its total kernel $K$ is a finite index subgroup of $\Z^A$ (\cite[Lemma 4.5]{part2}). In particular, it is generated as a group by $K \cap \N^A$. The nonnegative points in $K$ can be characterized as follows. 

\begin{lemma}
\cite[Lemma 4.8]{part2}
\label{l.free} 
If $\Net$ is finite and irreducible, then the following are equivalent for $\xx \in \N^A$.
	\begin{enumerate}
	\item $t(\xx)\qq = \qq$ for some locally recurrent $\qq$.
	\item $t(\xx)\qq = \qq$ for all locally recurrent $\qq$.
	\item $\xx \in K$.
	\end{enumerate}
\end{lemma}

The total kernel depends only on the state transition maps. Next we define the production map, which depends also on the message passing functions. Given a locally recurrent state $\qq$ and a vector $\kk \in K \cap \N^A$ we have	
	\begin{equation} \label{e.production} \kk \acts \qq = P_\qq(\kk).\qq \end{equation}
for some vector $P_\qq(\kk) \in \N^A$.  

Nonnegative elements of the total kernel can be thought of as reset vectors: if $\kk \in K \cap \N^A$ then processing $\kk_a$ letters $a$ for all $a \in A$ returns all processors to their initial (locally recurrent) states, and the vector $P_\qq(\kk)$  specifies how many letters of each type are passed as a result. For example, in the case of a sandpile network, the vectors $\{d_v 1_v\}_{v \in V}$ are an $\N$-basis for $K$, and the production map is $\N$-linear with
	\[ P_\qq (d_v 1_v) = \sum_{(v,u) \in E} 1_u. \]
This equation just says that the local action of $d_v$ letters $v$ is to send one letter $u$ along each outgoing edge $(v,u)$ and to return the processor at $v$ to its initial state.
Note that $P_\qq$ does not depend on $\qq$ in this example! The next lemma generalizes these two observations (linearity and independence of $\qq$) to a general abelian network. 

\begin{lemma} \cite[Lemmas 4.6 and 4.9]{part2}
Let $\Net$ be a finite abelian network. 
\begin{enumerate} 
\item $P_\qq : K \cap \N^A \to \N^A$ extends to a group homomorphism $K \to \Z^A$.
\item If $\Net$ is irreducible then $P_\qq = P_\rr$ for all $\qq,\rr \in Q$.
\end{enumerate}
\end{lemma}

In light of (2), when $\Net$ is irreducible we will often drop the subscript and denote the production map simply by $P$. In a slight abuse of notation, we also denote by $P$ the $A\times A$ matrix of the linear map $\Q^A \to \Q^A$ obtained by tensoring the production map with $\Q$. More explicitly, for any $x \in \mathbb{Q}^A$ there exists $n$ such that $nx \in K$ and we set $P(x) := (1/n)P(nx)$. The \emph{Laplacian} of $\Net$ is defined as the $A\times A$ matrix \begin{equation} \label{e.laplacian} L = (I-P)D \end{equation} where $I$ is the $A \times A$ identity matrix and $D$ is the diagonal matrix with diagonal entries
	\begin{equation}\label{e.ra}  r_a = \min \{ m \geq 1 \,:\, m1_a \in K \}. \end{equation}
This set is nonempty because $K$ is a finite index subgroup of $\Z^A$. For example, in a sandpile network (\S\ref{s.examples}) we have $r_a = d_a$, the outdegree of vertex~$a$. In general one can think of $r_a$ as a ``reset number'': If $a \in A_v$ then inputting $r_a$ letters $a$ resets processor $\Proc_v$ to its initial (locally recurrent) state.

\begin{theorem} \cite[Theorem 5.6 and Corollary 6.4]{part2}
\label{t.halting}
Let $\Net$ be a finite irreducible abelian network $\Net$ with production matrix $P$ and Laplacian $L$. The following are equivalent.
\begin{enumerate}
\item $\Net$ halts on all inputs.
\item The spectral radius of $P$ is strictly less than $1$.
\item All principal minors $L$ are positive.
% needed for det L>0 later
\end{enumerate}
\end{theorem}

The proof in \cite{part2} uses Dickson's Lemma and the Perron-Frobenius theorem. We sketch here just the equivalence of (1) and (2). If $P$ has Perron-Frobenius eigenvalue $\lambda \geq 1$, then the corresponding eigenvector is used to construct an input for which $\Net$ fails to halt. Conversely, if $\Net$ fails to halt on some input, then by Dickson's Lemma there exists a legal execution from some $\yy.\qq$ to some $\zz.\qq$ (with the \emph{same} state $\qq$) satisfying $\zz \geq \yy$. This implies $P(\yy) \geq \yy$ and hence $\lambda \geq 1$.

Regarding (3), note that $L$ need not be symmetric. A real square matrix with nonpositive off-diagonal entries and positive principal minors is called a \emph{toppling matrix} (or ``M-matrix'').  
The positive principal minor condition is one way to extend the notion of ``positive definite'' to non-symmetric matrices.  
Fiedler and Ptak \cite[Theorem 4.3]{FP62} list of thirteen equivalent conditions of which this is one. The classic example of a toppling matrix is the reduced Laplacian of a strongly connected 
% or sink-connected
directed graph. Extending theorems about graph Laplacians to toppling matrices is an ongoing topic of research \cite{PS04,GK15}.
%Toppling matrices are also known as ``M-matrices''.  

\section{Critical Group}
\label{s.critical}

The main results begin in this section.
Throughout this section, we take $\Net$ to be a finite irreducible abelian network that halts on all inputs. 

\subsection{Action on recurrent states}
\label{s.recurrent}

 Since the global monoid $M$ of \eqref{e.globalmonoid} is finite and commutative it has a minimal idempotent $e$, and $eM$ is a finite abelian group with identity element $e$ (see \textsection \ref{s.monoid}).

\begin{definition} 
\label{d.critical}
The \emph{critical group} $\Crit \Net$ is the group $eM$. 
\end{definition}

\begin{definition}
A state $x \in Q$ is \emph{recurrent} if it satisfies the equivalent conditions of Lemma~\ref{l.recurrent} (for example, $x=ex$) for the defining action $M \times Q \to Q$.  

Denote by $\Rec \Net$ the set of recurrent states of $\Net$.
\end{definition}

For example, if $\Net = \Rotor(G,s)$ is a simple rotor network on a directed graph $G$ with sink vertex $s$, then $\Rec \Net$ can be identified with spanning trees of $G$ oriented toward $s$ \cite[Lemma 3.16]{notes}. The critical group of $\Net$ is isomorphic to the sandpile group $\Crit \Sand(G,s)$ \cite{PDDK96, LL09}. We will deduce this isomorphism as a special case of Theorem~\ref{t.generalcriticalgroup}, below.

The following theorem generalizes \cite[Lemmas 3.13 and 3.17]{notes}, where it was shown that $\Crit \Rotor(G,s)$ acts freely and transitively on the set of spanning trees of $G$ oriented toward~$s$.
 	
\begin{theorem}
\label{t.freeandtransitive}
Let $\Net$ be a finite irreducible abelian network that halts on all inputs.  The action of the critical group on recurrent states
	\[ \Crit \Net \times \Rec \Net \to \Rec \Net \]
is free and transitive.  In particular, $\# \Crit \Net = \# \Rec \Net$.
\end{theorem}

\begin{proof}
Let $\M$ be the transition monoid of $\Net$, and let $e$ be the minimal idempotent of $M$.  Then $\Crit \Net = e\M$ and $\Rec \Net = eQ$.  The monoid action $\M \times Q \to Q$ is faithful by definition and irreducible by Lemma~\ref{l.localglobalirreducible}.  Hence the group action $e\M \times eQ \to eQ$ is free and transitive by Theorem~\ref{t.monoidtogroup}.
\end{proof}

%\begin{example}
\subsection{Markov chain}
\label{s.markov}

Next we formalize one way in which $\Crit \Net$ and its action on $\Rec \Net$ govern the ``long term behavior'' of $\Net$.  Let $\alpha$ be a probability distribution on the total alphabet $A$.
Consider the Markov chain $(\qq_n)_{n \geq 0}$ on state space $Q$ where the initial state $\qq_0$ can be arbitrary, and subsequent states are defined by
	\begin{equation} \label{e.themarkovchain} \qq_{n+1} = 1_{a_n} \Acts \qq_n, \qquad{ n\geq 0} \end{equation}
where the inputs $a_n \in A$ for $n\geq 0$ are drawn independently at random with distribution $\alpha$. 
% More generally, a_n itself could be a Markov chain with stationary distribution $\alpha$

Recall that $\qq$ is called \emph{recurrent} for the Markov chain if \begin{equation} \label{e.returnprob1} \Prob ( \qq_n = \qq \text{ for some } n \geq 1 \,|\, \qq_0=\qq ) = 1. \end{equation}
Next we relate this notion of recurrence to the algebraic notion. Let $M_\alpha$ be the submonoid of $M$ generated by $\{\tau_a \,:\, \alpha(a)>0\}$. We say that $\qq$ is \emph{accessible} from $\qq'$ if $\qq \in M_\alpha \qq'$. By the elementary theory of Markov chains, \eqref{e.returnprob1} holds if and only if 
	\begin{equation} \label{e.access} \text{$\qq$ is accessible from $m \qq$ for all $m \in M_\alpha$} \end{equation}
 (here we are using the assumption that $Q$ is finite). Mutual accessibility is an equivalence relation on $Q$. Its equivalence classes are called \emph{communicating classes}.
 
 Write $e_\alpha$ for the minimal idempotent of $M_\alpha$. 
 
 \begin{lemma}
 \label{l.ealpha}
 $\qq$ is recurrent for the Markov chain \eqref{e.themarkovchain} if and only if $e_\alpha \qq = \qq$. 
 \end{lemma}
 
 \begin{proof}
Given $m \in M_\alpha$, let $g$ be the inverse of $e_\alpha m$ in the group $e_\alpha M_\alpha$. Since $g = g e_\alpha$ we have $gm = g (e_\alpha m) = e_\alpha$, so for any $\qq \in Q$
	\[ g(m\qq) = e_{\alpha} \qq. \]
Therefore if $e_\alpha \qq = \qq$ then \eqref{e.access} holds. Conversely, if \eqref{e.access} holds then in particular $\qq$ is accessible from $e_\alpha \qq$, so $\qq \in e_\alpha Q$ and hence $e_\alpha \qq = \qq$.
 \end{proof}

If the support of $\alpha$ is too small, it may happen that $e_\alpha \neq e$,
or it may happen that $e_\alpha = e$ but $e M_\alpha$ is a proper subgroup of $eM$.
We say that $\alpha$ has \emph{adequate support} if
	\[ e_\alpha M_\alpha = eM. \]
%This condition implies $e_\alpha = e$ and a little bit more (in general it may happen that $e_\alpha = e$ but $eM_\alpha$ is a proper subgroup of $eM$). 
Note that this holds trivially if $\alpha(a)>0$ for all $a \in A$ (in which case $M_\alpha = M$).

\begin{lemma}
\label{l.markov}
 If $\alpha$ has adequate support, then $\Rec \Net$ is the unique communicating class of recurrent states for the Markov chain \eqref{e.themarkovchain}.
\end{lemma}

\begin{proof}
If $e_\alpha M_\alpha \subseteq eM$ then $e_\alpha = e$. By Lemma~\ref{l.ealpha} the set of recurrent states for the Markov chain is $e_\alpha Q = eQ = \Rec \Net$.

By Theorem~\ref{t.freeandtransitive}, the group $eM$ acts transitively on $\Rec \Net$, so if $eM \subseteq e_\alpha M_\alpha$ then any two states of $\Rec \Net$ are mutually accessible.
\end{proof}

 The next two theorems generalize results of \cite{Dha90}, where they are proved for sandpile networks.
 
  \begin{theorem}\label{c.markov}
For any $\alpha$, the uniform distribution on $\Rec \Net$ is stationary for the Markov chain \eqref{e.themarkovchain}. If $\alpha$ has adequate support, then the stationary distribution is unique.
  \end{theorem}

\begin{proof}
Fix $\qq_0 \in \Rec \Net$, and let $g$ be uniform random element of $\Crit \Net$.  By Theorem~\ref{t.freeandtransitive}, $g \Acts \qq_0$ is a uniform random element of $\Rec \Net$.  If $a \in A$ is independent of $g$, then $\tau_a g$ is also a uniform random element of $\Crit \Net$; hence if $\qq_n$ is uniform on $\Rec \Net$, then $\qq_{n+1}$ is again uniform on $\Rec \Net$.  

If $\alpha$ has adequate support, then $\Rec \Net$ is the unique recurrent communicating class by Lemma~\ref{l.markov}, so the stationary distribution is unique.
\end{proof}

\subsection{Expected time to halt}
\label{s.expectedtime}

Recall our assumptions that $\Net$ is finite and irreducible and halts on all inputs. In particular, the production matrix $P$ is well defined and has spectral radius $<1$ by Theorem~\ref{t.halting}, so $I-P$ is invertible where $I$ is the $A \times A$ identity matrix. Our next result gives an interpretation for the entry $(I-P)^{-1}_{ab}$: it is the expected number of letters $a$ processed before the network halts, when one letter $b$ is input to a uniform recurrent state. 

%, which does not depend on the choice of initial state since $\Net$ is assumed irreducible. 
%
%Since we assume $\Net$ irreducible, its production matrix $P$ does not depend on the choice of initial state; and since we assume $\Net$ halts on all inputs, $I-P$ is a toppling matrix by Lemma~\ref{l.mfrompositive} and Theorem~\ref{t.halting}. According to Lemma~\ref{l.posdef}, the inverse of a toppling matrix has nonnegative entries. 

\begin{theorem} 
\moniker{Expected Time To Halt}
\label{t.greenfunction}
Let $\qq$ be a uniform random element of $\Rec \Net$.  Then for all $\xx \in \N^A$ we have \[ \EE [\xx.\qq] = (I-P)^{-1} \xx. \]
\end{theorem}

Later, in Theorem~\ref{t.timetohalt}, we will bound the difference $[\xx.\qq] - \EE[\xx.\qq]$.
%The proof uses a few lemmas, which will also be used later in \textsection\ref{s.burning}.  
To build up to the proofs of these results, consider the group homomorphism
	\[ \phi: \Z^A \to \Crit \Net \]
defined on generators by $\one_a \mapsto e\tau_a$ for each $a \in A$.  For $x \in \N^A$ and $\qq \in Q$ we have
	\begin{equation} \label{e.phinonnegative} \phi(\xx)\qq = \left( \prod_{a \in A} (e\tau_a)^{\xx_a} \right) \qq = \left( \prod_{a \in A} \tau_a^{\xx_a} \right) e\qq = \xx \Acts e\qq. \end{equation}
First we observe that input vectors in the kernel of $\phi$ act trivially on recurrent states.	

\begin{lemma}\label{l.phikernel}
If $\xx \in (\ker \phi ) \cap \N^A$ and $\qq \in \Rec \Net$, then $\xx \Acts \qq = \qq$.
\end{lemma}

\begin{proof}
For $\qq \in \Rec \Net$ we have $\qq = e\qq$.
% by Lemma~\ref{l.recurrent}(5).  
Now by \eqref{e.phinonnegative}, since $\phi(\xx)=e$,
	\[ \xx \Acts \qq = \xx\Acts  e\qq = \phi(\xx)\qq = e\qq = \qq. \qedhere \]
\end{proof}

Write $e_v$ for the minimal idempotent of the local monoid $M_v$ (\textsection\ref{s.localmonoid}).

\begin{lemma}
\moniker{Recurrent Implies Locally Recurrent}
\label{l.locallyrecurrent}
If $e\qq = \qq$ then $e_v \qq_v = \qq_v$ for all $v \in V$.
\end{lemma}

\begin{proof}
By Lemma~\ref{l.largee} there exists $\zz \geq \one$ be such that $\tau(\zz)=e$. Let $\uu = [\zz.\qq] \geq \one$. For each $v \in V$ let $\yy_v \in \N^{A_v}$ be such that $t_v(\yy_v)=e_v$. Then for all sufficiently large $n$ we have $n \uu_v \geq \yy_v$ and hence $t_v(n \uu_v) \in e_v M_v$. Since $e_v M_v$ is a finite group, there exists $n_v \in \N$ such that $t_v(n_v \uu_v) = e_v$.

Now if $e\qq=\qq$ then for all $n \in \N$ we have $(n\zz) \Acts \qq = \qq$ and $[n\zz.\qq] = n\uu$ by Lemma~\ref{l.odomsplit}, and hence $t(n\uu)\qq = \qq$ by \eqref{e.tglobal}. Taking $n=n_v$ we obtain
	\[ e_v \qq_v = t_v(n_v \uu_v) \qq_v = \qq_v. \qedhere \]
\end{proof}

The converse of Lemma~\ref{l.locallyrecurrent} is usually false: For example, all states of $\Rotor(G,s)$ and $\Sand(G,s)$ are locally recurrent, but states of the former containing an oriented cycle of rotors are not recurrent, nor are states of the latter containing ``forbidden'' subconfigurations (the simplest of which is a pair of adjacent vertices both in state $0$). Theorem~\ref{t.nocycle} gives conditions when the converse does hold.

Now we come to the main ingredient in the proof of Theorem~\ref{t.greenfunction}: if the network starts in a locally recurrent state $\qq$ and halts in the same state $\qq$, then the odometer $[\xx.\qq]$ belongs to the total kernel $K$.

\begin{lemma}
\label{l.burningodom}
% note: we do not need to assume \qq recurrent, but we do need locally recurrent!
Fix $\xx \in \N^A$ and a locally recurrent state $\qq \in Q$.  Let $\kk = [\xx.\qq]$.  If $\xx \Acts \qq = \qq$, then $\kk \in K$ and
	$ \xx = (I-P)\kk. $
\end{lemma}

\begin{proof}
If $\xx \Acts \qq = \qq$, then $t(\kk) \qq = \qq$ by equation \eqref{e.tglobal}. 
By Lemma~\ref{l.free}
%, since $\qq$ is locally recurrent 
it follows that $\kk \in K$.
Now by Lemma~\ref{l.odometertrick}(ii) and the definition \eqref{e.production} of the production matrix,
	\[ \kk.\qq = \kk \acts (\xx.\qq) = (P(\kk)+\xx).\qq \]
%where in the second equality we have again used that $\qq$ is locally recurrent. 
Hence $\kk = P(\kk)+\xx$.
\end{proof}

\begin{proof}[Proof of Theorem~\ref{t.greenfunction}]
Since $\Crit \Net$ is a finite group, for any $\xx \in \Z^A$ there is a positive integer $n$ such that $n \xx \in \ker \phi$.
%Given $\xx \in \N^A$, choose a positive integer $n$ such that $n\xx \in \ker \phi$. 
Fix $\qq \in \Rec \Net$ and let $\kk = [n\xx.\qq]$. We have $n\xx \Acts \qq = \qq$ by Lemma~\ref{l.phikernel}. Moreover $\qq$ is locally recurrent by Lemma~\ref{l.locallyrecurrent}, so by Lemma~\ref{l.burningodom} it follows that $n\xx = (I-P)\kk$.  In particular, since $I-P$ is invertible, $\kk$ does not depend on $\qq$.  

Now by Lemma~\ref{l.odomsplit},
	\[ \kk = \sum_{j=0}^{n-1} [\xx.(j\xx \Acts \qq)]. \]
% the left side is deterministic, but each term on the right side is random.
By Theorem~\ref{c.markov}, if $\qq$ is uniform on $\Rec \Net$ then $j\xx \Acts \qq$ is uniform on $\Rec \Net$ for all $j \in \N$.  Taking expectations, we obtain
	\[ \kk = \EE\kk = n \EE [\xx.\qq]. \]
Dividing by $n$ yields the result.
\end{proof}

\subsection{Generators and relations}
\label{s.genrel}

In this section we give generators and relations for the critical group.  The group homomorphism $\phi: \Z^A \to \Crit \Net$
sending $a \mapsto e\tau_a$ is surjective, since $\Crit \Net = eM$ and $M$ is generated by $\{\tau_a\}_{a \in A}$.  To describe the kernel of $\phi$ we will use the \emph{production map} from \textsection\ref{s.KP},
	\[ P : K \to \Z^A \]
where $K \subset \Z^A$ is the total kernel. Write $I$ for the inclusion $K \hookrightarrow \Z^A$.

\begin{theorem}
\label{t.generalcriticalgroup}
The natural map $\phi : \Z^A \to \Crit \Net$ induces an isomorphism of abelian groups
	\[ \Crit \Net \isom \Z^A / (I-P)K. \]
\end{theorem}

\begin{proof}
We must show that $\ker \phi = (I-P)K$.  Fix $\qq \in \Rec \Net$.  For any $\kk \in K \cap \N^A$ we have
	\[ \kk \acts \qq = P(\kk).\qq \]
so $\kk \Acts \qq = P(\kk) \Acts \qq$ by Lemma~\ref{l.partialexec2}.  
Hence $\phi(\kk)\qq = \phi(P(\kk))\qq$ by \eqref{e.phinonnegative}.
By Theorem~\ref{t.freeandtransitive} the action of $\Crit \Net$ on $\Rec \Net$ is free, % since q is arbitrary we are really only using faithful, not free
so we conclude $\phi(\kk) = \phi(P(\kk))$.  This shows that $(I-P)\kk \in \ker \phi$ for all $\kk \in K \cap \N^A$. 
%By Lemma~\ref{l.fullrank}, 
Since $K$ is generated as a group by $K \cap \N^A$, it follows that $(I-P)K \subset \ker \phi$.

To show the reverse inclusion, given $\xx \in (\ker \phi) \cap \N^A$ we have $\phi(\xx)=e$ and hence $\xx \Acts e \qq = e \qq$ for all $\qq \in Q$ by \eqref{e.phinonnegative}.  By Lemma~\ref{l.burningodom}, $\xx = (I-P) \kk $ where $\kk = [\xx.e\qq] \in K$.  It follows that $(\ker \phi) \cap \N^A \subset (I-P)K$.  Since $\Crit \Net$ is finite, $\ker \phi$ is a full rank subgroup of $\Z^A$, so it is generated as a group by $(\ker \phi) \cap \N^A$, which shows that $\ker \phi \subset (I-P)K$.
\end{proof}

Irreducible abelian networks $\Net$ and $\Net'$ on the same graph with the same total alphabet are called \emph{homotopic}, written $\Net \approx \Net'$, if they have the same total kernel $K$ and the same production map $P$. An example of homotopic networks are the sandpile and simple rotor networks on a connected graph $G=(V,E)$. These have total kernel $K = \prod_{v \in V} (d_v \Z)$ where $d_v$ is the degree of vertex $v$: in the sandpile case processing $d_v$ letters $v$ causes $v$ to topple exactly once, while in the rotor case it causes the rotor at $v$ to serve each neighbor $u$ of $v$ exactly once.  In both cases exactly one letter is sent to each neighbor, so the production map is given by $P(\kk)_u = \sum_{(v,u) \in E} \kk_v/d_v$.

By Theorem~\ref{t.generalcriticalgroup} the critical group depends only on the total kernel $K$ and production matrix $P$, so homotopic networks have isomorphic critical groups.

\begin{corollary}
\label{c.homotopy}
If $\Net \approx \Net'$, then $\Crit \Net \simeq \Crit \Net'$.  
\end{corollary}

Corollary~\ref{c.homotopy} generalizes the isomorphism $\Crit \Rotor(G,s) \simeq \Crit \Sand(G,s)$ between the rotor and sandpile groups of a graph with sink vertex $s$.

Next we relate the critical group to the cokernel of the Laplacian \eqref{e.laplacian}.

\begin{definition}\label{d.rectangular}
An abelian network $\Net$ is \emph{rectangular} if its total kernel is $K = \prod_{a \in A} (r_a \Z)$, where $r_a$ are the reset numbers \eqref{e.ra}. (That is, $K$ is a rectangular sublattice of $\Z^A$.)
\end{definition}

\begin{corollary}
\label{c.laplaciancokernel}
The natural map $\Z^A \to \Crit \Net$ induces a surjective group homomorphism
	\[ \phi: \Z^A / L\Z^A \twoheadrightarrow \Crit \Net. \]
If $\Net$ is rectangular, then $\phi$ is an isomorphism.
\end{corollary}

\begin{proof}
By definition, $D\Z^A \subset K$ with equality if $\Net$ is rectangular.
Since $L = (I-P)D$, we have $L \Z^A \subset (I-P)K$ with equality if $\Net$ is rectangular.
\end{proof}

Note that any unary network (and in particular any toppling network) is rectangular. In \cite{part2} we defined the \emph{sandpilization} $\Sa(\Net)$ as the locally recurrent toppling network with the same Laplacian as $\Net$. To spell this out, let us first describe the underlying graph of $\Sa(\Net)$, which has vertex set $A$ rather than $V$. 

\begin{definition}
\label{d.production}
The \emph{production graph} $\Gamma(\Net)$ is the directed graph with vertex set $A$ and edge set $\{(a,b) \,:\, P_{ba}>0\}$.
\end{definition}

To form the sandpilization, we split each processor $\Proc_v$ into $\# A_v$ unary processors $\ProcR_a$ for $a \in A_v$.  Each $\ProcR_a$ has state space $\{0,1,\ldots,r_a-1\}$ where $r_a$ is the reset number \eqref{e.ra}. When $\ProcR_a$ transitions from state $r_a-1$ to state $0$, it sends $r_a P_{ba}$ letters $b$ to each $\ProcR_b$.  Note that $r_a P_{ba}= P(r_a \basis_a)_b$ is an integer since $r_a \basis_a \in K$.

\begin{definition}
\label{d.sandpilization}
The \emph{sandpilization} of $\Net$ is the unary network 
	\[  \Sa(\Net) = (\ProcR_a)_{a \in A} \]
with underlying graph $\Gamma(\Net)$.
\end{definition}

Since $\Net$ and $\Sa(\Net)$ have the same Laplacian, the following is immediate from Corollary~\ref{c.laplaciancokernel}.

\begin{corollary}
\label{c.sandpilizationgp}
$\Crit \Sa(\Net) \isom \Z^A / L\Z^A$ and 
$\Crit \Sa(\Net) \twoheadrightarrow \Crit \Net$.
\end{corollary}

\paragraph{\textbf{Example.}}
To see that the map $\Crit \Sa(\Net) \to \Crit \Net$ need not be an isomorphism, consider the following non-rectangular network with vertices $i,j$ where $j$ is a sink. Let $Q_i=\{0,1\}$, $A_i=\{a,b\}$, $A_j=\{c\}$, 
\[ T_i(q,a) = T_i(q,b) =q+1 \modnospace{2} \]
and
\[ T_{(i,j)}(0,a)=c \]
\[ T_{(i,j)}(1,a)=cc \]
\[ T_{(i,j)}(0,b)=\eps \]
\[ T_{(i,j)}(1,b)=c. \]
Figure~\ref{f.nonrectangular} shows the state diagram of $\Proc_i$.

\begin{figure}
\begin{center}
\includegraphics[scale=1]{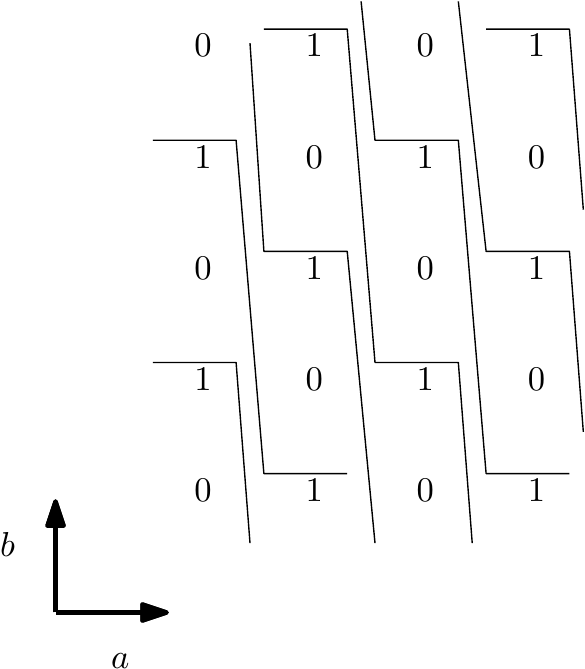}
\end{center}
\caption{State diagram of an abelian processor with two states $0,1$ and two inputs $a,b$. Its kernel is $\{(m,n) \in \Z^2 \,:\, m+n \equiv 0 \pmod 2 \}$, so it is not rectangular. Each line crossed results in output of one letter $c$.}
\label{f.nonrectangular}
\end{figure}

\begin{comment}
\begin{figure}
\begin{center}
\includegraphics[scale=1]{nonrectangular_example_diagram2.pdf}
\end{center}
\end{figure}
\end{comment}

The production matrix and Laplacian of this network, with rows and columns indexed by $a,b,c$ in that order, are: 
\[ P=\left( \begin{array}{ccc}
0 & 0 & 0 \\
0 & 0 & 0 \\
\frac{3}2& \frac{1}2 & 0 \end{array} \right)
\qquad
L=\left( \begin{array}{ccc}
2 & 0 & 0 \\
0 & 2 & 0 \\
-3 & -1 & 1 \end{array} \right)\] 
We have $\Crit \Sa (\Net) = \Z^3 / L\Z^3 = (\Z/2\Z)^2$. On the other hand, $\tau_a = \tau_b$ in $\Crit \Net$ (both send $0 \mapsto 1 \mapsto 0$) so $\Crit \Net=\Z/2\Z$.

Processor $\Proc_i$ in this example has another curious feature: The input words $aa,ab,bb$ to state $0$ all result in the same sequence of states $0,1,0$ yet they produce different outputs ($ccc$, $cc$, $c$ respectively). $\qed$

\begin{comment}
We now consider a slight variation $\Net'$ of $\Net$. It consists of two vertices, with $j$ as before, but $i$ sends slightly different messages:
 
$$T_i(n,x)=n+1 \pmod{2}, \hspace{.1in} x \in A_i$$
and for some vertex $j$, with $c\in A_j$,
$$T_{(i,j)}(0,a)=c  $$
$$T_{(i,j)}(1,a)=b,c  $$
$$T_{(i,j)}(0,b)=\emptyset  $$
$$T_{(i,j)}(1,b)=b $$

The Laplacian of $\Net'$ is 
\[ \left( \begin{array}{ccc}
2 & 0 & 0 \\
-1 & 1 & 0 \\
-2 & 0 & 1 \end{array} \right)\] 
so $\Crit \Sa (\Net)=\Crit \Net=\Z_2$. In particular, it can happen that the isomorphism of Theorem~\ref{t.laplaciancokernel} holds even if $\Net$ is not rectangular.
\end{comment}

\subsection{Order of the critical group}
\label{s.order}

Now we turn to the problem of counting recurrent states, or equivalently (by Theorem~\ref{t.freeandtransitive}) finding the order of the critical group.  Let 
	\[ \iota = [K : D \Z^A] \]
be the index of $D \Z^A$ as a subgroup of $K$.  Recalling that $K =\prod_{v \in V} K_v$, we can write $\iota = \prod_{v \in V} \iota_v$ as a product of local indices
	\[ \iota_v = [K_v : D_v \Z^{A_v}] \]
where $D_v \Z^{A_v} := \prod_{a \in A_v} (r_a \Z)$.  Note $\iota = 1$ if and only if $\Net$ is rectangular.

\begin{theorem}\label{t.recurrentcount} 
	\[ \# \Rec \Net = \# \Crit \Net = \frac{\det L}{\iota} . \]
\end{theorem}

\begin{proof}
The first equality follows from Theorem~\ref{t.freeandtransitive}.  For the second, we have by Theorem~\ref{t.generalcriticalgroup}
	\begin{align*} \# \Crit \Net = [\Z^A : (I-P)K ]
		&= \frac{[\Z^A : L\Z^A]}{[(I-P)K : L\Z^A]} \\
		&= \frac{\left|\det L \right|}{[(I-P)K : (I-P)D\Z^A]}.
		 \end{align*}
By Theorem~\ref{t.halting}, since $\Net$ halts on all inputs, $I-P$ has full rank and $\det L>0$, so the right side equals $\frac{\det L}{[K: D\Z^A]}$.
\end{proof}

\begin{comment}
\begin{example}
Let $\Net$ be a height arrow network (\textsection \ref{s.heightarrow}) with sink $s$.  This is a unary network, so it is rectangular.
%Its total kernel is $K = D \Z^V$ where $r_s = 1$ and $r_v = \lcm(\tau_v,d_v)$
%%= \frac{\tau_v d_v}{g_v}$.  
%for $v \neq s$.
Each local component has the same production matrix $p_{uv} = d_{uv}/d_v$, and $D$ has diagonal entries $r_s=1$ and $r_v = \lcm(\tau_v,d_v)$ for $v\neq s$, so the Laplacian is $L = (I-P)D = \Delta^s \Lambda$, where $\Delta^s$ is the damped graph Laplacian and $\Lambda$ is the diagonal matrix with diagonal entries $\lambda_s=1$ and $\lambda_v = r_v/d_v$ for $v \neq s$. 
By Corollary~\ref{c.laplaciancokernel}, each local component $\Net_\qq$ has $\Crit \Net_\qq \simeq \Z^V / L \Z^V$, and the number of recurrent states in each local component is $\det L = (\prod \lambda_v) \det \Delta^s$.  The number of local components is $\prod g_v$, where  $g_v = \gcd(\tau_v,d_v)$. Noting that $\lambda_v = \tau_v/g_v$, the total number of recurrent states is $(\prod \tau_v) \det \Delta^s$.
This formula agrees with \cite[Proposition 3.9]{DR04}; note that the last line there has a misprint: $d_i$ should be $\tau_i$.
\end{example}
\end{comment}

Our last goal in this section is to address the converse of Lemma~\ref{l.locallyrecurrent}: When does locally recurrent imply recurrent? We will find several equivalent conditions.

\begin{lemma}
\label{l.chiponcycle}
Let $\mathcal{S}$ be a toppling network with production graph $\Gamma$. For each directed cycle $a_1 \rightarrow \ldots \rightarrow a_m \rightarrow a_1$ of $\Gamma$ we have
	\[ \max \{ \qq(a_1),\ldots,\qq(a_m) \} \geq 1 \]
for all $\qq \in \Rec \mathcal{S}$. 
\end{lemma}

\begin{proof}
Let $\rr = (r_a)_{a \in A}$ be the vector of toppling thresholds. By Lemma~\ref{l.largee} there exists $\zz \geq \rr$ such that $\tau(\zz)=e$. If $\qq \in \Rec \mathcal{S}$ then
	$ \zz \Acts \qq = \qq $
and each vertex topples at least once during the stabilization of $\zz.\qq$.  If $a_i$ is the last vertex on the cycle to finish toppling, then $a_{i+1}$ receives a chip from $a_i$ and does not topple thereafter, so $\qq(a_{i+1}) \geq 1$.
\end{proof}

\begin{theorem}
\label{t.nocycle}
Let $\Net$ be a finite irreducible abelian network that halts on all inputs, and let $\Sa(\Net)$ be its sandpilization.
The following are equivalent.
	\begin{enumerate}
	\item[(1)] Every locally recurrent state of $\Net$ is recurrent.
	\item[(2)] $\det L = \det D$.
	\item[(3)] Every state of $\Sa(\Net)$ is recurrent.
	\item[(4)] The state $\zero$ of $\Sa(\Net)$ is recurrent.
	\item[(5)] The production graph $\Gamma$ has no directed cycles.
	\item[(6)] The producrtion matrix $P$ is nilpotent.
	\end{enumerate}
\end{theorem}

\begin{proof}
Recall from \textsection\ref{s.KP} the action of $\Z^{A_v}$ on $e_v Q_v$ for each vertex $v$.
Since $\Net$ is irreducible this action is transitive, and its kernel is $K_v$. So the number of locally recurrent states of $\Net$ is \[ \prod_{v \in V} \# (e_v Q_v) = \prod_{v\in V} \# (\Z^{A_v}/K_v) = [\Z^A : K]. \]
By Theorem~\ref{t.recurrentcount}, the number of recurrent states of $\Net$ is \[ \frac{\det L}{[K : D\Z^A]} = \frac{\det L}{\det D} [\Z^A : K]. \]  
Now from Lemma~\ref{l.locallyrecurrent} it follows that (1) $\Leftrightarrow$ (2).

All states of the sandpilization $\Sa(\Net)$ are locally recurrent, and $\Sa(\Net)$ has the same matrices $L$ and $D$ as $\Net$, so (2) $\Rightarrow$ (3).

Trivially (3) $\Rightarrow$ (4).

If $\Gamma$ has a directed cycle, then the state $\zero$ of $\Sa (\Net)$ is not recurrent by Lemma~\ref{l.chiponcycle}, which shows (4) $\Rightarrow$ (5).

If $\Gamma$ has no directed cycles, then each entry of $P^k$ is a weighted sum over directed paths in $\Gamma$ with $k+1$ distinct vertices. Taking $k=\# A$ we obtain $P^k=\zero$, which shows (5) $\Rightarrow$ (6).

If $\lambda_1, \ldots, \lambda_n$ are the eigenvalues of $P$ with multiplicity, then
	\[ \frac{\det L}{\det D} = \det (I-P) = \prod_{i=1}^n (1-\lambda_i). \]
If $P$ is nilpotent then $\lambda_1=\cdots=\lambda_n=0$, so $\det L= \det D$. Hence (6) $\Rightarrow$ (2), completing the proof.
\end{proof}

\section{Time to halt}
\label{s.timetohalt}

In this section we show that the production matrix of $\Net$ determines its running time on any input up to an additive constant (Theorem~\ref{t.timetohalt}). In the special case of rotor networks, additive error bounds of this type appear in the work of Cooper and Spencer \cite{CS06} and Holroyd and Propp \cite{HP10}.  For an upper bound in the case of sandpiles, see \cite[Prop.\ 4.8]{notes}. A related but distinct question is: given $\xx$ and $\qq$, how quickly can one compute the final state $\xx \Acts \qq$ and the odometer $[\xx.\qq]$?  According to Theorem~\ref{t.timetohalt} the abelian network $\Net$ performs this computation (asynchronously) in time approximately $\one^T (I-P)^{-1}\xx$, but in some cases \cite{BS13,FL13} one can design a (sequential) algorithm that is much faster.

As usual we take $\Net$ to be a finite irreducible abelian network that halts on all inputs. Let $K$ be its total kernel and $P : K \to \Z^A$ its production map. Write $I$ for the inclusion $K \hookrightarrow \Z^A$. 

\begin{lemma}\label{l.countmessages}
Suppose $\kk \in K$ satisfies $P\kk  \leq \kk$.  If $\qq \in \Rec \Net$, then 
	\[ (I-P)\kk\Acts\qq=\qq. \] 
%and
%	\[ [(I-P)\kk.\qq] = \kk. \]
\end{lemma}

\begin{proof}
Let $\xx = (I-P)\kk \in \N^A$.
By Theorem~\ref{t.generalcriticalgroup}, if $\kk \in K$ then $\xx \in \ker \phi$.  The conclusion now follows from Lemma~\ref{l.phikernel}. 
\end{proof}

\begin{remark} Recalling that $I-P$ is injective (Theorem~\ref{t.halting}) and that the matrix of $(I-P)^{-1}$ has nonnegative entries (Theorem~\ref{t.greenfunction}), we see that the condition $P\kk \leq \kk$ in the above lemma implies $\kk \geq \zero$.
%
%In the next theorem and its proof, we write $(I-P)^{-1}\xx$ for the image of $\xx \in \Z^A$ under this matrix (even when $\xx \not \in (I-P)K$).  
\end{remark}

For $\uu \in \Z^A$ write $\| \uu \|_{\infty} = \max_{a \in A} |\uu_a|$.

\begin{theorem}
\label{t.timetohalt}
\moniker{Time To Halt}
Let $\Net$ be an irreducible finite abelian network that halts on all inputs.  
There is a constant $C$ depending only on $\Net$, such that
for all $\xx \in \N^A$ and all $\qq \in Q$,
	\[  \big\| \, [\xx.\qq] - (I-P)^{-1}\xx \, \big\|_{\infty}  \leq C.\]
\end{theorem}

\begin{proof}
Since $(I-P)K$ has full rank in $\Z^A$, there is a constant $c$ such that for any $\xx \in \N^A$ there exists
$\yy = (I-P)\kk \in (I-P)K$ such that $\zero \leq \yy-\xx \leq c\one$. Since $(I-P)^{-1}$ has nonnegative entries, 
	\[ \zero \leq \kk - (I-P)^{-1}\xx \leq (I-P)^{-1} c \one. \]
Therefore it suffices to bound $\| [\xx.\qq]-\kk \|_{\infty}$.  

To do this, fix $\zz \in \N^A$ with $\tau(\zz) = e$. We will show
	\begin{equation} \label{e.C1} \|[\xx.e\qq] - \kk\|_\infty \leq C_1 \end{equation}
and
	\begin{equation} \label{e.C2} \|[\xx.e\qq] - [\xx.\qq] \|_\infty \leq C_2 \end{equation}
where $C_1 =  \max_{\qq' \in Q} \| [c\one.\qq'] \|_{\infty}$ and $C_2 = \max_{\qq' \in Q} \| [\zz.\qq'] \|_{\infty}$. 

Since $\yy \geq \zero$, we have $P\kk \leq \kk$, so $\yy \Acts e\qq = e\qq$ by Lemma~\ref{l.countmessages}.  Hence by Lemmas~\ref{l.burningodom} and~\ref{l.odomsplit}
	\[ \kk = [\yy.e\qq] = [\xx.e\qq] + [(\yy-\xx).(\xx \Acts e\qq)] \]
which shows \eqref{e.C1}. Two more applications of Lemma~\ref{l.odomsplit} give
	\[ [\zz.\qq] + [\xx.e\qq] = [(\xx+\zz).\qq] = [\xx.\qq] + [\zz.(\xx \Acts \qq)] \]
which shows \eqref{e.C2}.
\end{proof}

By Theorem~\ref{t.timetohalt} and the triangle inequality, \[ \| [\xx.\qq] - [\xx.\rr] \|_{\infty} \leq 2C \] for all $\xx \in \N^A$ and all $\qq,\rr \in Q$. We will use this bound in the next section.

\section{Burning test}
\label{s.burning}

Dhar's burning test \cite{Dha90} is an efficient algorithm for determining whether a state of $\Sand(G,s)$ recurrent. In its simplest form it applies to Eulerian directed graphs $G$ (strongly connected, $\operatorname{indegree}(v)=\operatorname{outdegree}(v)$ for $v\in V$); see \cite[\textsection 4]{notes}.  Speer \cite{Spe93} treated general directed graphs.  A variant of Speer's algorithm can be found in \cite{PPW11}.  The goal of this section is to generalize these algorithms to the setting of an arbitrary finite irreducible abelian network $\Net$ that halts on all inputs.  

\begin{definition}
A \emph{burning element} for $\Net$ is a vector $\beta \in \N^A$ such that for $\qq \in Q$ 
	\[ \qq \in \Rec \Net  \quad \Leftrightarrow \quad  \beta \Acts \qq = \qq. \] 
\end{definition}

To see that burning elements exist, we can use Lemma~\ref{l.recurrent}(5): $\qq \in \Rec \Net$ if and only if $\qq=e\qq$, where $e$ is the minimal idempotent of the transition monoid $M$. Recalling the map $\tau$ of \eqref{e.taudef}, any $\beta \in \N^A$ such that $\tau(\beta)=e$ is a burning element. However, such elements are typically large. The power of the burning test derives from the fact that one can often identify a small burning element $\beta$ to reduce the running time $[\beta.\qq]$.

In the classical setting of $\Net = \Sand(G,s)$, where $G$ is a directed graph with globally accessible sink $s$,
there is a pointwise minimal burning element $\beta$.
In the case that $G$ is Eulerian, 
the usual formula for the burning element is $\beta_v = d_{vs}$, the number of edges into $v$ from $s$. The special role of the sink is undesirable for us (since in general, an abelian network need not have a sink in order to halt on all inputs).  To remove it, notice that \[ \beta = L\one \] where $L$ is the Laplacian of \eqref{e.laplacian}.
 (To make the connection to the graph Laplacian $\Delta_G = D_G - A_G$, where $D_G$ is the diagonal matrix of outdegrees, and $A_G$ is the adjacency matrix of $G$, we have $L = \Delta_{G'}$ where $G'$ is the graph obtained from~$G$ by removing all outgoing edges from $s$.  Since $G$ is Eulerian we have $\Delta_G \one = \zero$, so $L\one = - \Delta_G \one_s = \beta$.)

When $G$ is not Eulerian, $L\one$ may have negative entries.  Speer \cite{Spe93} observed that there is a pointwise smallest vector $\yy \geq \one$ such that $L\yy \geq \zero$, and showed that $\beta = L\yy$ is a burning element. Our burning test for abelian networks, Theorem~\ref{t.burning}, reduces to Speer's in the case $\Net = \Sand(G,s)$.

\subsection{Large inputs}
\label{s.large}

To lay the ground for the burning test, we show in this section that states obtained from sufficiently large inputs to $\Net$ must be recurrent. We consider two interpretations of ``large'': either the number of letters in the input $\xx$ is large (Lemma~\ref{l.largeinput}), or the odometer $[\xx.\qq]$ is large (Lemmas~\ref{l.largeodoms} and~\ref{l.largeodom}). 

\begin{lemma}
\label{l.largeinput}
Let $\zz \in \N^A$ be such that $\tau(\zz)=e$. If $\xx \geq \zz$ then $\xx \Acts \qq$ is recurrent for all $\qq \in Q$.
\end{lemma}

\begin{proof}
Writing $\xx = \yy+\zz$ for $\yy \in \N^A$, we have
	\[ \xx \Acts \qq = \zz \Acts (\yy \Acts \qq) = e m \qq. \]
where $m = \tau(\yy)$. By Lemma~\ref{l.recurrent}(4) the state $em\qq$ is recurrent.  
\end{proof}

Given $\qq,\rr \in Q$ we say that $\qq$ is \emph{locally accessible} from $\rr$ if $\qq = \mm \rr$ for some $\mm \in \prod_{v \in V} M_v$. In particular, if there is an execution (say $w$) from $\xx.\rr$ to $\yy.\qq$, then $\qq$ is locally accessible from $\rr$ (namely $\qq = t(|w|)\rr$).

\begin{lemma}
\label{l.largeodoms}
If there exists $\xx \in \N^A$ such that $\xx \Acts \rr = \rr$ and $[\xx.\qq] \geq \one$ for all states $\qq$ locally accessible from $\rr$, then $\rr$ is recurrent.
\end{lemma}

\begin{proof}
Let $\zz \in \N^A$ be such that $\tau(\zz)=e$. 
We will find a $\yy \geq \zz$ and a state $\qq$ such that $\rr = \yy \Acts \qq$.

Let $n = \sum_{a \in A} \zz_a$.  Fix any sequence $a_1,\ldots,a_n$ such that $\zz = \sum_{i=1}^n \one_{a_i}$.  Let $\qq^0=\rr$ and inductively define states $\qq^1,\ldots,\qq^n$, all locally accessible from $\rr$, as follows.

For each $i=1,\ldots,n$, since $\qq^{i-1}$ is locally accessible from $\rr$ we have $[\xx.\qq^{i-1}]_{a_i} \geq 1$, so there is a legal execution $w_i$ from $\xx.\qq^{i-1}$ to some state $\yy^i.\qq^i$ satisfying $\yy^i_{a_i} \geq 1$.  The concatenation $w_1\cdots w_n$ is a legal execution from $(n\xx).\rr$ to $\yy.\qq^n$ with $\yy = \yy^1+\cdots+\yy^n \geq \zz$.  Hence by Lemma~\ref{l.partialexec},
	\[ \rr = (n\xx) \Acts \rr = \yy \Acts \qq^n \]
and the right side is recurrent by Lemma~\ref{l.largeinput}.
\end{proof}

For $\uu \in \N^A$ and $\qq \in Q$, define	
	\[ R_\uu(\qq) = \Set{\xx \Acts \qq}{\xx \in \N^A,\, [\xx.\qq] \geq \uu}. \]
and write $R_\uu$ for the set of states $\qq \in Q$ such that $\qq \in R_\uu(\qq$).

We have defined recurrent states in Lemma~\ref{l.recurrent} by a list of equivalent monoid-theoretic properties.  Now we can add to this list a characterization that is specific to abelian networks.  
According to the next lemma, a state $\qq$ is recurrent if and only if it there exist inputs with arbitrarily large odometers that fix $\qq$.
% not as obvious as it seems, because the definition of recurrent is "attainable from arbitrarily large inputs". Odometers that are large everywhere can be produced by inputs that are not large everywhere (e.g. many letters input to one vertex). 

\begin{lemma}
\label{l.largeodom}
	\[  \Rec \Net = \bigcap_{\uu \in \N^A} R_\uu. \]
\end{lemma}

\begin{proof}
Suppose that $\rr \in \bigcap R_\uu$.  Then for any $\uu \in \N^A$ there is an input $\xx \in \N^A$ with $\rr = \xx \Acts \rr$ and $[\xx.\rr] \geq \uu$.  Take $\uu = (2C+1)\one$. By Theorem~\ref{t.timetohalt} it follows that $[\xx.\qq] \geq \one$ for all $\qq \in Q$.  Now from Lemma~\ref{l.largeodoms} it follows that $\rr$ is recurrent.

It remains to show that $\Rec \Net \subset R_\uu$ for all $\uu \in \N^A$.
Since $\Crit \Net$ is a finite group, the kernel of $\phi : \Z^A \to \Crit \Net$ has nonempty intersection with $\N^A + \uu$.  Let $\xx$ be a point in this intersection.  Then for any $\qq \in \Rec \Net$ we have  $\xx \Acts \qq = \qq$ by Lemma~\ref{l.phikernel}.  Since $[\xx.\qq] \geq \xx \geq \uu$, it follows that $\qq \in R_\uu$.
\end{proof}

Finally, we show that if $\kk \geq \one$ then the converse to Lemma~\ref{l.countmessages} holds. 

\begin{theorem}\label{t.burning}
\moniker{Burning Test}
Let $\kk\in K$ be such that $\kk \geq \one$ and $P\kk \leq \kk$. Then $\qq\in Q$ is recurrent if and only if $(I-P)\kk \Acts \qq=\qq$. 
\end{theorem}

\begin{proof}
If $\qq$ is recurrent, then $(I-P)\kk \Acts \qq = \qq$ by Lemma~\ref{l.countmessages}.

For the converse, let $\xx = (I-P)\kk$ and $\uu = [\xx.\qq]$.  Suppose that $\xx \Acts \qq = \qq$.  Then $\xx = (I-P)\uu$ by Lemma~\ref{l.burningodom}.  Since $I-P$ is injective, we obtain $\uu=\kk$.  Now for any $n \in \N$ we have $n\xx \Acts \qq = \qq$ and $[n\xx.\qq] = n\kk$.  Hence
	\[ \qq \in \bigcap_{n \geq 1} R_{n\kk}. \]
Since $\kk \geq \one$ the	right side equals $\bigcap_{\uu \in \N^A} R_\uu$, which equals $\Rec \Net$ by Lemma~\ref{l.largeodom}.
\end{proof}

\begin{remark}
The preceding theorem can be improved slightly by weakening $\kk \geq \one$ to $\kk \geq \one_C$, where $C$ is the set of all $a \in A$ that lie on a directed cycle of the production graph $\Gamma$ 
(Definition~\ref{d.production}).  
The reason is that if $a \in A$ does not lie on a directed cycle, then $\{a\}$ is a strong component of $\Gamma$, and all locally recurrent states of the strong component $\Net^a$ are recurrent.  Writing $a^\omega$ for the minimal idempotent of the transition monoid of $\Net^a$, one checks that $e = \prod_{a \in C} a^\omega$ and hence that $\Rec \Net = \bigcap_{a \in C} \Rec \Net^a$. See \cite{C+13} where this improvement is carried out in detail for sandpile networks.
\end{remark}

\subsection{Finding a burning element}
\label{s.finding}

In this section we show that there is always a burning element $\beta$ satisfying $\zero \leq \beta \leq \rr$, where $\rr_a$ is the smallest positive integer such that $\rr_a \one_a \in K$.

\begin{definition}
A \emph{burning odometer} is a vector $\kk$ satisfying
	\begin{equation} \label{e.burningconstraints} (I-P)\kk \ge \zero, \qquad  \kk\ge \one, \qquad  \kk\in K.  \end{equation}
The corresponding \emph{burning element} is $\beta = (I-P)\kk$.
\end{definition}

According to Theorem~\ref{t.burning}, to check whether $\qq$ is recurrent it suffices to find a burning odometer $\kk$ and then compute $\beta \Acts \qq$, where $\beta = (I-P)\kk$. The proof also shows that if $\qq$ is recurrent, then the local run time of this computation is $[\beta.\qq] = \kk$
(that is, for each $a \in A$ it requires processing $\kk_a$ letters $a$).  
Therefore we are interested in finding a burning odometer $\kk$ as small as possible.

Recall that $K \subset D\Z^A$, with equality if $\Net$ is rectangular.
For a sandpile network $\Net = \Sand(G,s)$, a burning odometer is $\kk = D\yy$ where $\yy$ is the ``burning script'' of \cite{Spe93}. More generally, if $\Net$ is rectangular, then writing $\kk = D\yy$, \eqref{e.burningconstraints} is equivalent to
\begin{equation} \label{e.burningprogram} L\yy \ge \zero, \qquad \yy\ge \one, \qquad \yy\in \Z^A \end{equation}
where $L= (I-P)D$ is the Laplacian of $\Net$.  Setting $\xx = \yy-\one$ we find that \eqref{e.burningconstraints} is equivalent to
	\begin{equation} \label{e.burningminusone} L\xx \geq -L\one, \qquad \xx \geq \zero, \qquad \xx \in \Z^A. \end{equation}
Minimizing $\one^T \xx$ subject to these constraints is an integer program of the class solved by toppling networks \cite[Remark~4.9]{part1}. Specifically, consider the sandpilization $\Sa(\Net)$, enlarged to allow negative chip counts. In this network, vertex $a$ has toppling threshold $\rr_a$, where $\rr$ is the vector of diagonal entries of the diagonal matrix $D$. 
By Theorem~\ref{t.halting}, if $\Net$ halts on all inputs then $\Sa(\Net)$ halts on all inputs (since $\Sa(\Net)$ has the same Laplacian as $\Net$).  For a chip configuration $\qq \in \Z^A$, write $\qq^\circ$ for the stabilization of $\qq$ in $\Sa(\Net)$.

\begin{corollary} 
\label{t.burningrect}
\moniker{$\Sa(\Net)$ computes a minimal burning element for $\Net$}
Let $\Net$ be an irreducible rectangular network that halts on all inputs. Then $\Net$ has a pointwise smallest burning odometer $\kk = D\yy$.  The corresponding burning element is given by 
	 \[ \beta = L\yy = \rr -\one - (\rr - \one - L\one)^\circ. \]
%	\[ = (D-I)\one - ((PD-I)\one)^\circ \]
Moreover, $\zero \leq \beta \leq \rr$.
\end{corollary}

\begin{proof}
Let $\qq = \rr-\one -L\one$, and for $a \in A$ let $\xx_a$ be the number of times $a$ topples during the stabilization of $\qq$ in $\Sa(\Net)$.  By \cite[Remark~4.9]{part1}, $\xx$ is the pointwise smallest vector satisfying \eqref{e.burningminusone}, and $\qq^\circ = \qq - L\xx$.  Setting $\yy = \xx+\one$, we conclude that $\kk = D\yy$ is the pointwise smallest vector satisfying \eqref{e.burningconstraints}, and
	\[ \beta = (I-P)\kk = L(\xx + \one) = L\one + \qq - \qq^\circ  = \rr - \one - \qq^\circ. \]
Noting that $L\one \leq \rr$, we have $-\one \leq \qq^\circ \leq \rr-\one$ and hence $\zero \leq \beta \leq \rr$.
\end{proof}

\begin{remark}
%Note that $(L\one)_a=\rr_a$ only if $a$ has no outgoing arcs in $\Gamma$.  In this case $a \not\in C$.  
If we replace $\one$ by $\one_C$ in \eqref{e.burningconstraints}, then we obtain a slightly smaller burning element $\beta$ which satisfies $\zero \leq \beta \leq \rr-\one$.  To compute this smaller element using $\Sa(\Net)$, take $\qq = \rr-\one-L\one_C$.  
\end{remark}

In practice, it is more direct to find the burning element by ``untopplings'' instead of topplings, which amounts to the following procedure to find the minimal $\yy$ satisfying \eqref{e.burningprogram}.

\begin{proc}\label{ILPsoln}
Start with $\yy=\one$.  If $L\yy \geq \zero$, then stop. Otherwise, choose some $a \in A$ such $(L\yy)_a<0$ and increase $\yy_a$ by $1$. Repeat until $L\yy \geq \zero$. 
\end{proc}
%\begin{remark}
%At the end of this procedure, we have $(L\yy)_a\le \rr_a$: this holds for $\yy=\one$, and each step preserves this condition.
%\end{remark}

\begin{remark}In the case that the row sums of $L$ are nonnegative, the procedure halts immediately with $\yy=\one$. In particular, this includes the special case $\Net=\Sand(G,s)$ for an Eulerian graph $G$. 
\end{remark}

In the case $\Net$ is not rectangular, the inclusion $K \subset D\Z^A$ is strict. Corollary~\ref{t.burningrect} and Procedure~\ref{ILPsoln} will identify the minimal burning odometer $ D\yy \in D\Z^A$.  Unlike the rectangular case, there may not be a unique minimal burning odometer in $K$.

If $\kk = D\yy$ is the minimal burning odometer in $D\N^A$, then the global burning test $(I-P)\kk \Acts \qq$ runs in time $\kk$.  Since $\beta = (I-P)\kk \leq \rr$, an upper bound for this run time is $\kk \leq (I-P)^{-1} \rr$. 
%	\[ \kk = (I-P)^{-1} \beta \leq (I-P)^{-1} \rr. \]

\section{Concluding Remarks}

We conclude with a few directions for future research.

\silentsubsec{Combinatorics of recurrent states}

The recurrent states of the rotor network $\Rotor(G,s)$ are the oriented spanning trees of $G$ rooted at $s$.
The recurrent states of the sandpile network $\Sand(G,s)$ on an undirected (or Eulerian directed) graph $G$ have a characterization in terms of ``forbidden subconfigurations'' \cite{Dha90} which puts them naturally in bijection with the $G$-parking functions of Postnikov and Shapiro \cite{PS04}.  Recently Guzm\'{a}n and Klivans \cite{GK15} have generalized this correspondence to toppling networks. 
Hopkins and Perkinson \cite{HP14} relate the $G$-parking functions to the Pak-Stanley labeling of a bigraphical hyperplane arrangement.
It would be interesting to find combinatorial characterizations of the recurrent states of other abelian networks. 

The rank function in Baker and Norine's Riemann-Roch theorem for graphs \cite{BN07} has the following interpretation: given a state of the sandpile network $\Sand(G)$ with no sink, what is the smallest number of letters that can be input to cause it to run forever?
Are there analogues of the Baker-Norine theorem and the Lorenzini zeta function \cite{Lor12} for more general abelian networks?

Duval, Klivans and Martin \cite{DKM09,DKM13} define a higher dimensional critical group as the cokernel (over $\Z$) of the Laplacian of a simplicial complex. Does this group have a dynamical interpretation in terms of a ``hypernetwork'' in which a set of nodes can interact if they form a face of the simplicial complex?  
 Can the spanning trees of a directed hypergraph, as defined by Gorodezky and Pak \cite{GP14}, be realized as the recurrent states of a suitable hypernetwork?
 
\silentsubsec{Critical networks}

Let us call an abelian network $\Net$ \emph{critical} if its production matrix has spectral radius $1$. By 
Theorem~\ref{t.halting}, 
a critical network has inputs that cause it to run forever, so it does not have a critical group in the sense of \textsection \ref{s.critical}.  However, one could try to define $\Crit \Net$ by generalizing the sinkless construction of the sandpile group: $ (\Z^V)_0 / L \Z^V$, where $(\Z^V)_0$ is the kernel of the map $\xx \mapsto \one^T \xx$.

For $\alpha, \beta \in \N^A \times Q$, write $\alpha \to \beta$ if there exists a legal execution from $\alpha$ to $\beta$. Let us call $\alpha$ \emph{recurrent} if for any $\beta$ such that $\alpha \to \beta$ we have $\beta \to \alpha$.
In a critical network $\Net$, is there an efficient test analogous to the burning algorithm to check whether $\alpha$ is recurrent?

In the case of a simple rotor network $\Rotor(G)$ with no sink, a state $\alpha = 1_v.\qq$ with just one letter is recurrent in this sense if and only if its rotors $(u,\qq_u)_{u \in V}$ form a \emph{cycle-rooted spanning tree} (a spanning subgraph with  a single oriented cycle) with $v$ lying on the cycle \cite[Theorem~3.8]{notes}.  Does this result extend to networks of abelian mobile agents (a non-unary generalization of rotors proposed in \cite{part1})?  What about states $\alpha$ with more than one letter? 

\silentsubsec{Finer algebraic invariants}

The critical group $\Crit \Net$ depends only on the homotopy type of $\Net$ (Corollary~\ref{c.homotopy}). On the other hand, the global monoid $M$ of \eqref{e.globalmonoid} can detect finer information about $\Net$.
To see that the monoid \emph{action} $M \times Q \to Q$ is not a homotopy invariant, note that the sandpile network $\Sand(G,s)$ has a state $\zero \in Q$ which can access all other states: $M\zero = Q$. If $G - \{s\}$ has two directed cycles that share an edge, then $\Rotor(G,s)$ has no such state, because a progressed cycle of rotors once broken can never be reformed.

Neither is $M$ itself a homotopy invariant: for instance, for the discrete torus $G = \Z/n \times \Z/n$ one can show that
the minimal burning element $\beta$ of $M(\Rotor(G,s))$ has $\beta^n = e$, whereas the corresponding exponent for $\beta$ in $M(\Sand(G,s))$ grows quadratically in $n$.

Recall that there are many distinct rotor networks $\Rotor(G,s)$ depending on the choice of ordering of the outgoing edges of each vertex.
An interesting question is whether $M$ can distinguish between these networks.

\section*{Acknowledgments}

This research was supported by an NSF postdoctoral fellowship and NSF grants DMS-1105960 and DMS-1243606, and by the UROP and SPUR programs at MIT. A few of the concluding remarks were inspired by discussions at the AIM workshop on generalizations of chip-firing and the critical group in July, 2013. The full list of open problems proposed at the workshop can be found at \url{http://aimath.org/WWN/chipfiring/aim_chip-firing_problems.pdf}.

\end{document}